\documentclass[a4paper,UKenglish,cleveref, autoref, thm-restate]{lipics-v2021}
\usepackage{amsmath,amsthm,amssymb}
\usepackage{microtype}
\usepackage{hyperref}
\usepackage{graphicx}
\usepackage[thinc]{esdiff}

\newcommand{\Bezier}{B{\'e}zier}
\renewcommand{\emph}[1]{\textbf{\textit{#1}}}
\usepackage{bmpsize}
\title{Drawing Planar Graphs and 1-Planar Graphs Using Cubic \Bezier{} Curves
with Bounded Curvature}
\titlerunning{Drawing Planar Graphs and 1-Planar Graphs Using Cubic \Bezier{} Curves}

\author{David Eppstein}{University of California, Irvine, USA \and \url{https://ics.uci.edu/~eppstein/}}{eppstein@uci.edu}{}{}

\author{Michael T. Goodrich}{University of California, Irvine, USA \and
\url{https://ics.uci.edu/~goodrich/}}{goodrich@uci.edu}{https://orcid.org/0000-0002-8943-191X}{}

\author{Abraham M. Illickan}{University of California, Irvine, USA \and \url{https://ics.uci.edu/~theory}}{aillicka@uci.edu}{https://orcid.org/0009-0006-4410-7098}{}

\authorrunning{Eppstein, Goodrich, and Illickan}

\Copyright{David Eppstein, Michael T. Goodrich, and Abraham M. Illickan} %TODO mandatory, please use full first names. LIPIcs license is "CC-BY";  http://creativecommons.org/licenses/by/3.0/

\ccsdesc[500]{Theory of computation}

\keywords{graph drawing, planar graphs, \Bezier~curves, and RAC drawings} %TODO mandatory; please add comma-separated list of keywords

\funding{This work was supported in part by NSF Grant 2212129.}%optional, to capture a funding statement, which applies to all authors. Please enter author specific funding statements as fifth argument of the \author macro.

\nolinenumbers %uncomment to disable line numbering

\begin{document}

\maketitle
\begin{abstract}
We study algorithms for drawing 
planar graphs and 1-planar graphs using cubic \Bezier{} curves with 
bounded curvature.
We show that any $n$-vertex 1-planar graph has a 
1-planar RAC drawing using a single cubic \Bezier{} curve per edge, 
and this drawing can be computed in $O(n)$ time given a combinatorial
1-planar drawing.
We also show that any $n$-vertex planar graph $G$ can be drawn in
$O(n)$ time with 
a single cubic \Bezier{} curve per edge, 
in an $O(n)\times O(n)$ bounding box, such that the edges
have $\Theta(1/{\rm degree}(v))$ angular
resolution, for each $v\in G$, and $O(\sqrt{n})$ curvature.
\end{abstract}

\section{Introduction}
A \Bezier{} curve is a parametric curve defined by a set of 
\emph{control} points that determine a smooth, 
continuous curve in the plane~\cite{duncan2005bezier,mortenson1999mathematics}.
For example, one of the most common types, a \emph{cubic \Bezier} curve,
is defined by four points, 
$P_0, P_1, P_2, P_3$, 
such that the curve starts at $P_0$ tangent to the line segment
$\overline{P_0P_1}$
and ends at $P_3$ tangent to the line segment $\overline{P_2P_3}$, with the 
lengths of $\overline{P_0P_1}$
and $\overline{P_2P_3}$ determining ``how fast'' the 
curve turns towards $P_1$ before turning towards $P_2$.
Formally, a cubic \Bezier{} curve, $f$, has the following
explicit form (see \cref{fig:bezier}):
\[
f(t)=(1-t)^{3}{P}_{0}+3(1-t)^{2}t{P}_{1}+3(1-t)t^{2}{P}_{2}+t^{3}P_{3},
\mbox{~~~for $0\leq t\leq 1.$}
\]

\begin{figure}[hbt]
\centering
\includegraphics[width=3.5in]{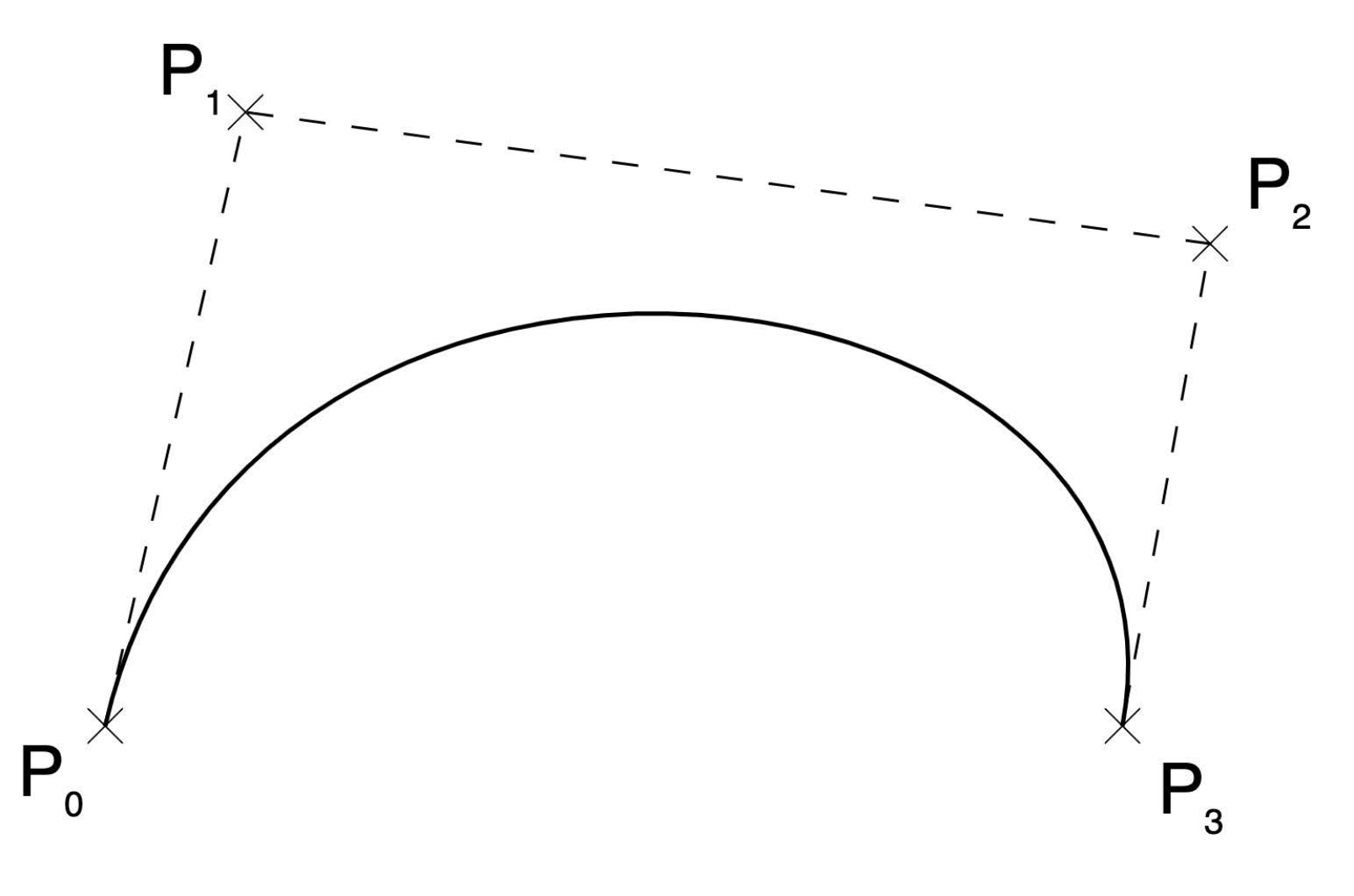}
\caption{An example cubic \Bezier{} curve. Public domain image by MarianSigler.}
\label{fig:bezier}
\end{figure}

The speed in which a curve turns can be characterized by its
\emph{curvature}, which is a measure of the the instantaneous rate 
of change of direction of a point that moves on the curve; hence,
the larger the curvature, the larger this rate of change.
For example, the curvature of a line is zero,
the curvature of a polygonal chain with a bend is $+\infty$, and
the curvature of a circle is the reciprocal of its radius.

Formally, the curvature of a twice-differentiable 
parameterized curve, $\mathbf{c}(t)=(x(t),y(t))$,
can be defined as follows (e.g., see~\cite[p.~890]{calc}):
\[
\kappa(t) = \frac{|x'y''-y'x''|}{(x'^2 + y'^2)^{3/2}},
\]
where $x'$ and $x''$ are the first and second derivatives of $x$, 
and $y'$ and $y''$ are the first and second derivatives of $y$, 
with respect to $t$. For this to be well-defined the curve must be smooth enough to have a second derivative, which is not true for polylines. In such cases the curvature can be thought of as infinite, as a limiting case of smooth perturbations of the given curve.

For a graph drawing, we desire the curvatures of the edges to be small,
where we define the curvature of a drawing of a graph $G$
to be the maximum curvature for a non-vertex point on an edge of $G$, taken
over all edges in the drawing of $G$.
For example, Xu, Rooney, Passmore, Ham, and Nguyen~\cite{xu12} 
empirically show that user performance 
on network tasks is better for low-curvature drawings than for
high-curvature drawings.

Unfortunately, minimizing curvature may conflict with other goals for 
a drawing.
For example, curvature can conflict with angular resolution
for planar drawings of planar graphs.\footnote{Recall that 
  a planar graph can be drawn in the plane without edge crossings
  and a 1-planar graph can be drawn in the plane so that
  each edge crosses at most one other edge. Also recall
  that the angular resolution for each vertex is the minimum angle between two
  edges incident on $v$ in the drawing.}
For example, we can draw an $n$-vertex planar graph without
crossings using straight edges 
(i.e., with curvature~0)
but this can cause angular resolution to be $O(1/n)$, even
for low-degree vertices, e.g., 
see~\cite{de1990draw,kant1996drawing,schnyder1990embedding},
or even worse, e.g., see Tutte~\cite{tutte}.
Indeed, Garg and Tamassia~\cite{garg} show that, in general, the 
best angular resolution that any algorithm 
for drawing a degree-$d$
planar graph $G$ using straight-line drawing can achieve
is $O(\sqrt{(\log d)/d^3})$.
Ideally, we would like the angular resolution for a drawing of a graph
$G$, to be $\Omega(1/{\rm degree}(v))$, for  each $v\in G$, which Goodrich
and Wagner show how to achieve~\cite{GOODRICH2000399}, but their methods for
achieving this bound either use polylines with bends (hence, with infinite
curvature) or with \Bezier{} curves that the authors admit have
high curvature, and they pose as an open problem whether one can
simultaneously achieve good angular resolution and relatively low
curvature for planar graph drawings with edges represented
with cubic \Bezier{} curves.\footnote{For the sake of normalization 
    of the curvature parameter, we assume in this paper that 
    a drawing has an $O(n)\times O(n)$ bounding box, as is common
    for drawings of planar and 1-planar graphs.}

In terms of another trade-off for drawings with curves,
Huang, Eades, Hong, and Duh~\cite{huang16}
empirically show that users performing network tasks 
were quickest with drawings with curved crossing edges rather than 
mixed drawings with no crossings, and the authors
conclude that
for better human graph comprehension, it might be better to use
curves to increase crossing angle, rather than to remove them completely.
Similarly,
Huang, Hong, and Eades~\cite{huang2008effects} 
report on user studies showing that crossings with large angles are 
much less harmful to the readability of drawings than shallow crossings.
Relatedly, there is considerable prior work on right angle crossing (RAC)
drawings, where every pair of crossing edges must cross at right angles,
but these drawings are typically achieved by using polygonal paths
with bends, e.g., 
see~\cite{ANGELINI202042,BEKOS201748,didimo2011drawing,didimo16,suzuki,toth23}; 
hence, these drawings can have unbounded curvature.
Therefore, we are interested 
in methods for producing RAC 
drawings of 1-planar graphs using curves having bounded
curvature, e.g., such as can be achieved with cubic \Bezier{} curves.

\Bezier{} curves are used extensively in computer graphics
applications, where it is common to concatenate \Bezier{} curves
together to form a composite \Bezier{} curve,
e.g., see~\cite{mortenson1999mathematics}.
As long as each connection point is collinear with its two adjacent 
control points, then the resulting composite \Bezier{}
curve will be $C^1$ continuous,
but it will not necessarily have continuous curvature.
In addition, such representations can be quite complex, depending on
the number of pieces used, and the curvature at connection
points might not be well-defined or, even if it exists, it might not
be easy to bound.
Thus, we are interested in this paper on studying drawings of planar
graphs and 1-planar graphs 
using cubic \Bezier{} curves
where each edge is represented with a single cubic \Bezier{} curve,
so that each edge has bounded curvature.
In the case of 1-planar graph drawings, we desire
edge crossings to be at right angles,
and in the case of planar graph drawings, we 
would like to simultaneously achieve good angular resolution and low curvature.

\subsection{Related Prior Work}
There is some notable previous work on using \Bezier{} curves
for graph drawing, which we review below, but we are not aware of previous
work on using \Bezier{} curves to draw planar graphs with low curvature
per edge and optimal angular resolution or for RAC drawings of 1-planar
graphs.\footnote{At a workshop affiliated with GD 2023 to celebrate the 60th 
   birthday of Beppe Liotta, Peter Eades advocated for more research
   on the topic of using \Bezier{} curves to draw graphs, including
   results involving curvature guarantees.}

In addition to the work cited above,
there is some interesting prior work on using
\Bezier{} curves in graph drawing systems.
For example, the Graphviz system of Gansner~\cite{gansner2009drawing}
can render edges using \Bezier{} curves.
Finkel and Tamassia~\cite{finkel} describe a force-directed graph drawing
implementation that uses \Bezier{} curves to render graph edges by integrating
control points into the force equations.
Brandes and Wagner~\cite{brandes98}
visualize railroad systems with some edges rendered using \Bezier{}
curves, and Fink, Haverkort, N{\"o}llenburg, Roberts, Schuhmann, and 
Wolff~\cite{fink13}
provide a similar type of system for drawing metro maps.
The GDot system of Hong, Eades, and Torkel~\cite{hong21} uses
\Bezier{} curves to draw edges in graphs visualized as dot paintings.
In addition, the CelticGraph system of Eades, Gr{\"o}ne, Klein,
Eades, Schreiber, Hailer, and Schreiber~\cite{eades23}
draws graphs using Celtic knots with edges represented
as \Bezier{} curves with limited curvature.

In terms of additional theoretical work,
Eppstein, Goodrich, and Meng~\cite{eppstein2007confluent} show how 
to draw confluent 
layered drawings using \Bezier{} curves that combine multiple edges,
and Eppstein and Simons provide a similar result for Hasse
diagrams~\cite{eppstein12}.
In addition, there is considerable prior work on Lombardi drawings, 
where edges are drawn using circular arcs, 
e.g., see~\cite{roman12,duncan2012lombardi,duncan2012planar,eppstein2012planar,eppstein12,kindermann2017lombardi}, which we consider 
to be related work even though circular arcs are not \Bezier{} curves.
Cheng, Duncan, Goodrich, and Kobourov~\cite{Cheng1999DrawingPG}
show how to draw an $n$-vertex planar graph $G$ with asymptotically
optimal angular resolution, $O(1/{\rm degree}(v))$, for each $v\in G$,
using 1-bend polylines or circle-arc chains, both of which have unbounded
curvature.

\subsection{Our Results}
In this paper, we show how to draw 1-planar graphs as RAC drawings 
using a single cubic \Bezier{} curve for each edge; 
hence, with bounded curvature.
We also show how to draw planar graphs in an $O(n)\times O(n)$ grid
with good angular resolution by rendering each edge using
a cubic \Bezier{} curve with $O(\sqrt{n})$ curvature.
Our methods involve careful constructions and proof techniques for 
proving bounded curvature results, which may be applicable in other
settings.

Our constructions are also based in part on refinements of the 
\emph{convex hull property} of \Bezier{} curves, which is that
every point of a \Bezier{} curve lies inside the convex hull of its defining
control points, e.g., see~\cite{duncan2005bezier}.
In our results, however, the convex
hull property is not sufficient, since the regions in which we 
desire \Bezier{} curves to traverse are more restrictive than just
the convex hulls of control points.
Moreover, the convex hull property says nothing about right-angle crossing
points for pairs of \Bezier{} curves, which is an important component
of our work, and one that requires considerable work, as we show.

\section{Constrained Constructions for Pairs of \Bezier{} Curves}
We show in this paper 
that we can draw any $1$-planar graph in the plane with right
angle crossings, i.e., a RAC drawing,
using \Bezier{} curves for every pair of intersecting
edges and straight line segments for the rest. 

Bekos, Didimo, Liotta, Mehrabi, and Montecchiani~\cite{BEKOS201748}
show that one can draw any 1-planar graph as a RAC drawing where 
each edge is represented by a polyline that has at most one bend.
Their algorithm starts from a 1-planar (combinatorial) embedding of a 
1-planar graph $G$, and proceeds via an induction proof involving
augmentation and contraction steps to produce a RAC drawing
of $G$ with edges represented with polylines with at most one bend
per edge.
We show how to adapt their proof to use a single cubic \Bezier{} 
curve per edge in place of a 1-bend polyline.
To achieve this, we develop a number of constructions for 
pairs of \Bezier{} curves that cross at right angles in specific ways
while fitting in specified polygonal regions.
As mentioned above, our constructions go well beyond the convex hull
property for \Bezier{} curves.
We describe each of these constructions in this section and we show
in the subsequent section how to use these constructions to prove 
our main result for RAC drawings of 1-planar graphs, which is the following.

\begin{restatable}{thm}{planarmain}
\label{thm:planarmain}
Any $n$-vertex $1$-planar graph has a $1$-planar RAC drawing using a single
cubic \Bezier{} curve per edge. 
Further, if a $1$-planar embedding of the graph has been provided, 
a $1$-planar RAC drawing using such cubic \Bezier{} curves 
can be computed in $O(n)$ time.
\end{restatable}
\addtocounter{theorem}{1}

Since \Bezier{} curves have bounded curvature, we achieve a RAC drawing of any 1-planar graph using edges
with bounded curvature.
We give each of our constructions for constrained pairs of \Bezier{}
curves in the subsections that follow.
Our constructions make use of another property of \Bezier{} 
curves; namely, that applying an affine transformation (e.g., 
a rotation, reflection, translation, or scaling) to a \Bezier{}
curve is equal to the \Bezier{} curve defined by applying that same 
transformation to the original control points, 
e.g., see~\cite{duncan2005bezier}.

\subsection{Right-angle Crossing in a Triangle and Outside a Quadrilateral}
Our first construction is for defining two cubic \Bezier{} curves
that have a right-angle crossing inside a triangle, each with two endpoints
that form a quadrilateral with the base of the triangle, such that
the curves lie outside of that quadrilateral. The curves that we describe are actually quadratic \Bezier{} curves but any quadratic \Bezier{} curve is also a cubic \Bezier{} curve by de Casteljau's algorithm~\cite{decasteljau59,decasteljau63}.
See \cref{fig:tri}.

\begin{figure}[hbt]
\centering
\includegraphics[width=5in]{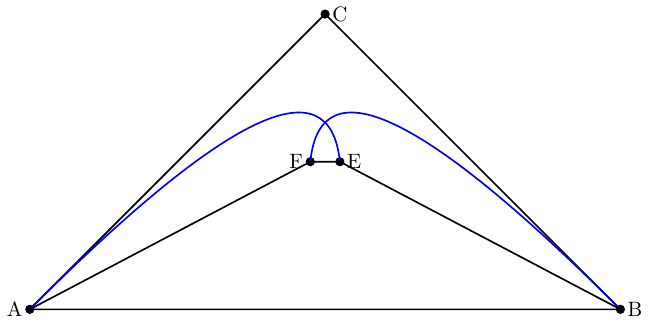}
\caption{Triangle $ABC$ with quadrilateral $ABEF$ and pair of cubic
\Bezier{} curves that cross in a right angle outside the quadrilateral
but inside $ABC$.}
\label{fig:tri}
\end{figure}

\begin{theorem}\label{outside curve}
For any triangle $ABC$, there is a quadrilateral $ABEF$ contained in $ABC$ such that there is a pair of \Bezier{} curves with pairs of endpoints $\{A,E\}$ and $\{B,F\}$ that intersect each other at right angles, are contained within $ABC$,
and lie outside of $ABEF$. The coordinates of $E$ and $F$ and the control points of the \Bezier{} curves can be computed efficiently, given the coordinates of $A$, $B$ and $C$.
\end{theorem}
\begin{proof}
It is enough to consider the case of an isosceles triangle 
with base $\overline{AB}$,
since such a triangle can be found within any given triangle $ABD$.
By the equivalence property for \Bezier{} curves under 
angle-preserving affine transformations,
let us assume $A=(0,0),B=(1,0),C=(1/2,C_y)$. Let $E=\left(\frac{1}{4} \left(4 C_y^2+6 C_y+\sqrt{16 C_y^4+48 C_y^3+40 C_y^2+12 C_y+1}+3\right),C_y/2\right)$ with $C_y>0$ and $F=(1-E_x,C_y/2) $. 
Let us define $g_1(t)=At^2+2Ct(1-t)+E(1-t)^2$ and $g_2(t)=Bt^2+2Ct(1-t)+F(1-t)^2$.  In the range $(0,1)$, these curves only meet at one point, 
\[
\left(\frac{1}{2},\frac{\left(-\sqrt{4 C_y^2+8 C_y+1}+6 C_y+1\right) \left(\sqrt{4 C_y^2+8C_y+1}+2 C_y-1\right)}{24 C_y}\right),
\]
when $t=\frac{4 C_y+1}{6 C_y}-\frac{\sqrt{4 C_y^2+8 C_y+1}}{6 C_y}$ for both curves.
    At this point, the slope of one curve is $-1$ and the other is $+1$, 
    so they intersect at a right angle.
    At $t=\frac{\sqrt{-4 E_x^2+8 E_x-3}-2 E_x+1}{2 (1-E_x)}$, $g_2$ is at 
\[
\left(E_x,-\frac{C_y \left(\sqrt{-4 E_x^2+8 E_x-3}-1\right) \left(3 \sqrt{-4 E_x^2+8 E_x-3}-8 E_x+5\right)}{8 (E_x-1)^2}\right),
\]
which is always higher than $E_y=C_y/2$. 
Thus, $g_2$ cannot intersect $\overline{EB}$ or $\overline{EF}$ except at $B$ and $F$ respectively because it is a parabola and $E$ is inside it. Similarly, $g_1$ does not intersect $\overline{AF}$ or $\overline{EF}$ except at $A$ and $F$ respectively.
\end{proof}

\subsection{Right-angle Crossing of a Diagonal of a Convex Quadrilateral}
Our next construction is for a cubic \Bezier{} curve $f$
that replaces a diagonal of a convex quadrilateral so that $f$
has a right-angle crossing with the other diagonal.
See \cref{fig:EFGHRAC}.

\begin{figure}[hbt]
    \centering
        \includegraphics[width=5in]{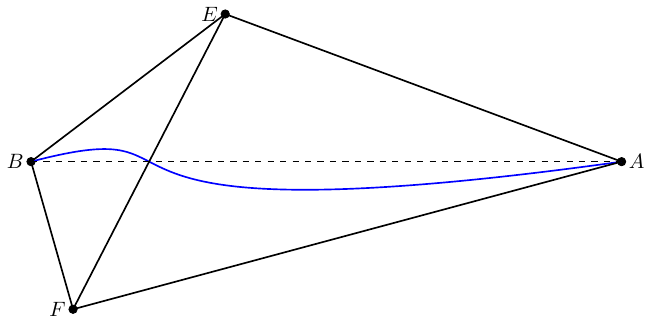}

    \caption{Pair of intersecting diagonals, $\overline{EF}$ and
    $\overline{AB}$, inside 
    a convex quadrilateral, $AEBF$. 
    We replace $\overline{AB}$ with a \Bezier{} curve that intersects 
    $\overline{EF}$ at a right angle.}
    \label{fig:EFGHRAC}
\end{figure}

We actually prove a slightly stronger result, for which, w.l.o.g., we 
assume the diagonal to be replaced is horizontal.

\begin{theorem}\label{insidecurve}
For any convex quadrilateral with horizontal diagonal, $\overline{AB}$, a point 
$X$ on $\overline{AB}$, and any real number, $m$, there is a simple cubic 
Bezier curve with endpoints $A$ and $B$ such that the \Bezier{} curve 
intersects $\overline{AB}$ at $X$ and makes slope $m$ 
at $X$ and the curve is contained in the quadrilateral and also stays in a pair of opposite quadrants around $X$. 
For $m\neq 0$, the curve intersects $\overline{AB}$ only at $A$, $X$, and $B$. The control points of the \Bezier{} curve can be computed efficiently, given the vertices of the quadrilateral.
\end{theorem} 

\label{bezierproof}
We will work with the points $A$ and $B$ being at $(1,0)$ and $(0,0)$ respectively. For any other pair of points we can apply an angle-preserving affine transformation. 
W.l.o.g.,
we also work with the point of intersection being on the line segment
from $(0,0)$ to $(8/9,0)$,
the portion of $\overline{AB}$ closer to $B$, 
by symmetry.
%First, we construct a method that works when the point of intersection has
%an $x$-coordinate in $[1/8,7/8]$. 
%Then we construct a method that works when the point of intersection 
%has an $x$-coordinate in $(0,((1+\sqrt{13})/6)^3)$. 
%These two ranges cover the interval $(0,7/8]$. 
%The first method is described in \cref{first curve intersecting}. 
%The second method is described in \cref{second curve intersecting}.
We will first show a curve that is perpendicular at $X$ in \cref{lem:1-right}. This first curve can be bounded by any arbitrary quadrilateral with $\overline{AB}$ as diagonal and also leaves two opposite quadrants around $X$ free for the other diagonal to be drawn as a straight line segment. 
Then, in \cref{lem:1-zero}, we show a curve that is a monotonic straight line that reaches $X$ 
at the same value of the parameter, $t$, as defines the first curve. 
(See~\cref{fig:cases}.)
Finally, in \cref{lem:1-conv}, we take a convex combination of these curves to get the required 
slope.
This third curve is the one that we use.

\begin{figure}[hbt]
\centering
\hspace{-15mm}
        \includegraphics[width=5in]{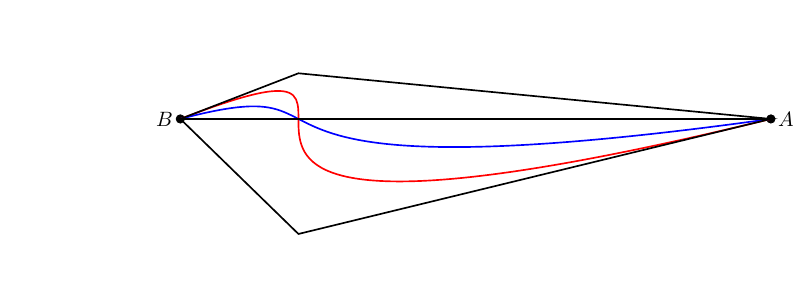}
\caption{
    \Bezier{} curves intersecting the line segment, $\overline{AB}$,
between the endpoints at different angles. 
The red curve is a curve which meets the line segment at that point at right angles. The blue curve is obtained by taking a convex combination with a curve that remains on the line segment and has the same parameter $t$ as the red curve at the point of intersection.
}
    \label{fig:cases}
\end{figure}

%\subsubsection{\Bezier{} curve intersecting at $t=1/2$ in $[1/8,7/8]$}\label{first curve intersecting}
%We first address the case for the point of intersection, $X$,
%having $x$-coordinate in $[1/8,7/8]$ and forming a right angle.
\begin{lemma}
\label{lem:1-right}
    Let $A=(1,0)$, $B=(0,0)$, and $X=(x_0,0)$ be a point on the line 
segment $AB$.
For any convex quadrilateral with $\overline{AB}$ as diagonal, there is 
a cubic \Bezier{} curve with $A$ and $B$ as endpoints and $C_1$ and $D_1$ as control points that intersects the 
line segment $\overline{AB}$ at $X$ at a right angle at the parameter
value $t=t_0=\frac{C_{1,x}-2D_{1,x}}{-1+3C_{1,x}-3D_{1,x}}$, and is contained within this quadrilateral and also within any pair of opposite quadrants around the point of intersection.
\end{lemma}
\begin{proof}
Assume $x_0<8/9$. Otherwise mirror the plane.
    Let $D_1=(D_{1,x},-r)$ and $C_1=(\frac{1}{2}(D_{1,x}-\sqrt{4D_{1,x}-3D_{1,x}^2}),\frac{1-2C_{1,x}+D_{1,x}}{2D_{1,x}-C_{1,x}}r)$ where $r\in \mathbb{R}$. We will define $D_{1,x}$ later. Let $f_1(t)=At^3+3C_1t^2(1-t)+3D_1t(1-t)^2+B(1-t)^3$ be the \Bezier{} curve
we are constructing. Then the $y$-coordinate function for $f_1$ is
    \begin{equation*}
        f_{1,y}(t)=-3\frac{1-2C_{1,x}+D_{1,x}}{2D_{1,x}-C_{1,x}}rt^2(1-t)+3rt(1-t)^2.
    \end{equation*}
    At $t=t_0$, $f_{1,y}(t)=0$. Note that there are three roots for the 
cubic polynomial $f_{1,y}$ and the other two are at $t=0$ and $t=1$.
The $x$-coordinate function for $f_1$ is
    \begin{equation*}
        f_{1,x}(t)=t^3+3\frac{1}{2}(D_{1,x}-\sqrt{4D_{1,x}-3D_{1,x}^2})t^2(1-t)+3D_{1,x}t(1-t)^2 .
    \end{equation*}
    At $t=t_0$, $f_{1,x}(t)=\frac{D_{1,x} \left(3 D_{1,x}+\sqrt{-D_{1,x} (3 D_{1,x}-4)}\right)}{3 D_{1,x}+3
   \sqrt{-D_{1,x} (3 D_{1,x}-4)}+2}$ and we choose $D_{1,x}$ to be $x_0+\sqrt[3]{x_0^2-x_0^3}$ which is a root of $f_{1,x}(t_0)=x_0$ when treated as an equation in $D_{1,x}$. This root is valid whenever $C_{1,x}$ is real, which happens when $4D_{1,x}-3D_{1,x}^2\geq 0$, which is true when $0\leq D_{1,x}\leq4/3$. This happens when $0\leq x_0\leq8/9$. Moreover,
\begin{eqnarray*}
        \diff{f_{1,y}}{t}=-\frac{3 r \left(C_{1,x} \left(9 t^2-8 t+1\right)+D_{1,x}
   \left(-9 t^2+10 t-2\right)+(2-3 t) t\right)}{C_{1,x}-2
   D_{1,x}} \mbox{~~~and}
\\
        \diff{f_{1,x}}{t}=3 \left(t (C_{1,x} (2-3 t)+t)+D_{1,x} \left(3 t^2-4
   t+1\right)\right).
\end{eqnarray*}
    At $t=t_0$, 
\[
\diff{f_{1,x}}{f_{1,y}}=\frac{\diff{f_{1,x}}{t}}{\diff{f_{1,y}}{t}}=0, 
\]
which means that the angle at the point of intersection is $\pi/2$. Also, the value of $C_{1,x}$ is such that 
$\diff{f_{1,x}}{t} = 3 t^2 (-3 C_{1,x}+3 D_{1,x}+1)+3 t (2 C_{1,x}-4 D_{1,x})+3 D_{1,x}$ is a quadratic polynomial in terms of $t$ with discriminant $0$. 
This means that $t_0$ is a repeated root and, hence, $f_{1,x}$ is monotonic. 
Thus, the curve remains in a pair of opposite quadrants, 
determined by the sign of $r$.
    Next, we show that this curve is bounded by the quadrilateral $(0,0),(x_0,\frac{-rx_0}{D_{1,x}}),(1,0),(x_0,\frac{r(1-2C_{1,x}+D_{1,x})(x_0-1)}{(C_{1,x}-1)(2D_{1,x}-C_{1,x})})$. This is obtained by the convex hull property of \Bezier{} curves in conjunction with the fact that $f_{1,x}$ is monotonic. For any value of $x_0$, these values are bounded.
    The value of $r$ can be adjusted such that this quadrilateral is contained within any quadrilateral with the diagonal $\overline{AB}$.
\end{proof}

We next address the case for the point of intersection $X$ 
having $x$-coordinate in $(0,8/9)$ and forming an angle of $0$.

% \subsubsection{Simple straight line curve that is at $(x_0,0)$ at $t=1/2$ for $x_0\in[1/8,7/8]$}\label{first curve flat}

\begin{lemma}
\label{lem:1-zero}
    Let $A=(1,0)$, $B=(0,0)$, and $X=(x_0,0)$ a point on the line segment $\overline{AB}$. There is a cubic \Bezier{} curve with $A$ and $B$ as endpoints that is at $X$ when $t=t_0$ as obtained from \cref{lem:1-right} for the same $X$ and monotonically traces the straight line segment $\overline{AB}$.
\end{lemma}
\begin{proof}
Assume $x_0<8/9$. Otherwise, mirror the plane.
    Let $C_2=D_2=(\frac{x_0-t_0^3}{3(1-t_0)t_0},0)$ and 
$f_2(t)=At^3+3C_2t^2(1-t)+3D_2t(1-t)^2+B(1-t)^3$ be the \Bezier{} curve
we are constructing.
    It is easy to see that for $f_2$'s $y$-coordinate function, $f_{2,y}=0$ 
throughout.
The $x$-coordinate function is
    \begin{equation*}
        f_{2,x}(t)=t^3+\frac{(1-t)t(x_0-t_0^3)}{(1-t_0)t_0}.
    \end{equation*}
    At $t=t_0$, $f_{2,x}(t)=x_0$. Also, $C_2$ and $D_2$, the repeated control points are between $A$ and $B$.
     
This means that the curve is always moving from 
$B$ to $A$ when $t$ goes from $0$ to $1$, 
without turning back or overshooting $A$.
\end{proof}

Given the previous two lemmas, we now can use them in convex combination.

% \subsection{A \Bezier{} curve intersecting at slope $m$ at $(x_0,0)$ at $t=1/2$ for $x_0\in[1/8,7/8]$}\label{first curve interpolated}

\begin{lemma}\label{lem:1-conv}
    Let $A=(1,0)$, $B=(0,0)$, $X=(x_0,0)$ a point on the line segment $\overline{AB}$ and $m$ a real number. 
For any convex quadrilateral with $\overline{AB}$ as a horizontal diagonal, 
there is a simple cubic \Bezier{} curve with $A$ and $B$ as endpoints that intersects the line segment $\overline{AB}$ at $X$ with slope $m$ 
for the parameter value $t=t_0$ as obtained from \cref{lem:1-right} for the same $X$ and is contained within this quadrilateral and also within any pair of opposite quadrants around the point of intersection.
\end{lemma}
\begin{proof}
    Let $C_3=kC_1+(1-k)C_2$ and $D_3=kD_1+(1-k)D_2$ for $k\in[0,1]$ where $C_1,D_1$ are control points of a curve $f_1(t)$ obtained from 
the proof of
\cref{lem:1-right} 
and $C_2,D_2$ are control points of a curve $f_2(t)$ obtained 
from the proof of
\cref{lem:1-zero}. 
Let 
\[
f_3(t)=At^3+3C_3t^2(1-t)+3D_3t(1-t)^2+B(1-t)^3=kf_1(t)+(1-k)f_2(t)
\]
be the more-general \Bezier{} curve we are now constructing as 
a convex combination of $f_1$ and $f_2$. 
Clearly, $f_3(t)$ is at $X$ when $t=t_0$.
In addition, we can write the slope of $f_3$ at $t_0$ as follows:
    \begin{equation*}
        \diff{f_{3,y}}{f_{3,x}} ~~=~~
       \frac{\diff{f_{3,y}}{t}}{\diff{f_{3,x}}{t}}
     ~~=~~ \frac{k\diff{f_{1,y}}{t}}{k\diff{f_{1,x}}{t}+(1-k)\diff{f_{2,x}}{t}}. 
    \end{equation*}
    To get a positive slope, we use a positive value of $r$. 
    To get a negative slope, we use a negative value of $r$. 
Values of $k$ in $[0,1]$ will span the full range of slopes
(either positive or negative) for every value of $r$.
    To get a specific slope $m$ at $t=t_0$, we set
    \begin{equation*}
        k=\frac{m\diff{f_{2,x}}{t}}{m(\diff{f_{1,x}}{t}-\diff{f_{2,x}}{t})+\diff{f_{1,y}}{t}} .
    \end{equation*}
    
    Now, we need to show that this curve does not self-intersect. $f_{3,x}$ is monotonic as it is a convex combination of monotonic functions. This means that no $x$-coordinate is repeated and the curve does not self intersect
The curve, $f_3$, is also bounded by the quadrilateral,
$(0,0),(x_0,\frac{-rx_0}{D_{1,x}}),(1,0),(x_0,\frac{r(1-2C_{1,x}+D_{1,x})(x_0-1)}{(C_{1,x}-1)(2D_{1,x}-C_{1,x})})$,
as it is a convex combination of the curves 
from \cref{lem:1-right} and \cref{lem:1-zero},
which are bounded by the same. 
The value of $r$ can be adjusted such that this quadrilateral 
is contained within any arbitrary quadrilateral. Again this curve also remains in a pair of opposite quadrants around $X$ since $f_{3,x}$ is monotonic.
\end{proof}

\section{RAC Drawings of 1-Planar Graphs with \Bezier{} Curves}\label{proof1planar}

We are now ready to prove 
\cref{thm:planarmain}, which is the following.

\planarmain*

% \begin{theorem}\label{planarmain}
% Any $n$-vertex $1$-planar graph has a $1$-planar RAC drawing using a single
% cubic \Bezier{} curve per edge. 
% Further, if a $1$-planar embedding of the graph has been provided, 
% a $1$-planar RAC drawing using such cubic \Bezier{} curves 
% can be computed in $O(n)$ time.
% \end{theorem}

The proof goes through three stages,
adapting a proof of
Bekos, Didimo, Liotta, Mehrabi, and Montecchiani~\cite{BEKOS201748} for
RAC drawings of 1-planar graphs with 1-bend polygonal edges.

\subsection{Augmentation}
We start with a $1$-plane combinatorial drawing $G$ of the graph. We call every connected region of the plane bounded by edges and parts of edges a face. The number of such edges or parts is called the length of the face. The induction will be using triangulated $1$-plane multigraphs, that is, $1$-plane multigraphs in which every face is of length $3$.
For every pair of crossing edges $ab,cd$, add edges $ac,cb,bd,ad$ such that the only edges contained within the cycle $acbd$ are $ab$ and $cd$. We call the subgraph consisting of these edges an empty kite. If there are $2$-length faces in this drawing, remove one of the edges recursively until there are no more. Also remove any parallel edge that was crossed. Now all faces in this drawing are either of length three with $2$ vertices and a crossing point or bounded only by vertices and no crossing points. In every face longer than $3$ add a new vertex and connect it to all the vertices on the face. We call this $1$-plane multigraph $G^+$. 
\subsection{Contraction}
A separation pair is a pair of vertices $\{u,v\}$ whose removal disconnects the graph. Lemma 5 in~\cite{BEKOS201748} states that between any separation pair $\{u,v\}$, there exist two parallel edges $e,e'$ such that $\{u,v\}$ is not a separation pair for the graph obtained by removing everything inside the cycle $\left<e,u,e',v\right>$. We call this removed subgraph along with the cycle $G_{uv}$. Replace $G_{uv}$ with a thick edge and iterate until there are no separation pairs. We call this graph $G^*$. $G^*$ is a simple triangulated $1$-plane graph.
\subsection{Drawing}
Obtain graph $H^*$ by removing all crossing pairs in $G^*$. All the faces of $H^*$ have either three or four vertices.
%\begin{lemma}
 %   $H^*$ is $3$-connected
%\end{lemma}
%\begin{proof}
%We see that $H^*$ is connected because only crossing pairs were removed and the rest of the empty kite is still present.\\ Suppose there were a cut vertex in $H^*$. Look at  a face that has more than one incidence of this vertex. In the walks from the first incidence to the second incidence and from the second incidence to the first incidence, there have to be at least two vertices each to avoid self loops, parallel edges and $1$-degree vertices. But the length of each face in $H^*$ is three or four. Therefore, $H^*$ is $2$-connected. \\
%Suppose there were a separation pair $\{u,v\}$ in $H^*$. Let $A$ be a component of $H^*\setminus\{u,v\}$. Let $f_1$ be the face that has the edge between $u$ and its neighbor in $A$ that is clockwise first after its neighbors not in $A$ and the edge between $v$ and its clockwise last neighbor in $A$ before its neighbors not in $A$. Let $f_2$ be the face that has the edge between $v$ and its neighbor in $A$ that is clockwise first after its neighbors not in $A$ and the edge between $u$ and its clockwise last neighbor in $A$ before its neighbors not in $A$. If both of these are of length three, $A\cup \{u,v\}$ is all the vertices. If one of them is of length three, $H^*$ has the edge $uv$. For each of them that is of length four a distinct copy of $uv$ would have been present in $G^*$. But $G^*$ is simple. $H^*$ cannot have any separation pairs.
%\end{proof}
Lemma 7 in~\cite{BEKOS201748} states that $H^*$ is $3$-connected. We can draw it with all faces convex and the outer face as a trapezoid or triangle using Tutte's method~\cite{tutte} or the method of Chiba {\it et al.}~\cite{chiba1984linear} to do it in linear time. Insert all the crossing edges in the interior four length faces by first drawing one of the edges as a straight line segment and then the other with the required slope using \cref{insidecurve}. If the outer face is of length four, use \cref{outside curve} with any large triangle.

For any thick edge $(u,v)$ with a triangle $uvx$ that $G_{uv}$ could be contained in (here $x$ need not be a vertex, it could be any point), consider $H_{uv}$ obtained by removing all crossing pairs. This is also $3$-connected like $H^*$. If the outer face of $H_{uv}$ is of length three and of the form $uvw$, recursively draw it inside $uvx$. If the outer face is of length four, use \cref{outside curve} to obtain the trapezoid and crossing pair and continue recursively.

\section{Drawing Planar Graphs Using \Bezier{} curves}
Suppose we are given an $n$-vertex planar graph $G$.
In this section,
we describe a method of drawing $G$ with asymptotically 
optimal angular resolution, $\Theta(1/{\rm degree}(v))$, for each $v\in G$,
using a single cubic \Bezier{} curve of curvature $O(\sqrt{n})$ for each edge. 
We start from an $O(n)\times O(n)$ grid drawing $D$ with asymptotically 
optimal angular resolution obtained by algorithm 
\texttt{OneBend} from Cheng, Duncan, Goodrich, 
and Kobourov~\cite{Cheng1999DrawingPG}, which uses one-bend edges.
We describe some properties of this drawing in \cref{onebendalgprop}. 
We describe our new drawing using Bezier curves in \cref{drawingdesc}. 
We show that the edges do not cross each other in \cref{sec:planarity}. We show that the vertices in our drawing have asymptocially optimal angular resolution in \cref{sec:angres}. We show that the curvature of our drawing is $O(\sqrt{n})$ in \cref{sec:curvature}.

\subsection{The Drawing Obtained by the \texttt{OneBend} Algorithm}
\label{onebendalgprop}
In the \texttt{OneBend} algorithm, $G$ is drawn 
in a drawing, $D$, in an $O(n)\times O(n)$ grid such that 
every vertex, $v$, has a \emph{joint box}---a square rotated $\pi/4$ of 
width and height $4{\rm degree}(v)+4$, centered at $v$.
Each joint box is divided into six regions. 
(See~\cref{fig:joint}.)
\begin{figure}[hbt]
\centering
\includegraphics[width=6in,trim=0.5in 1.5in 0in 1.5in, clip]{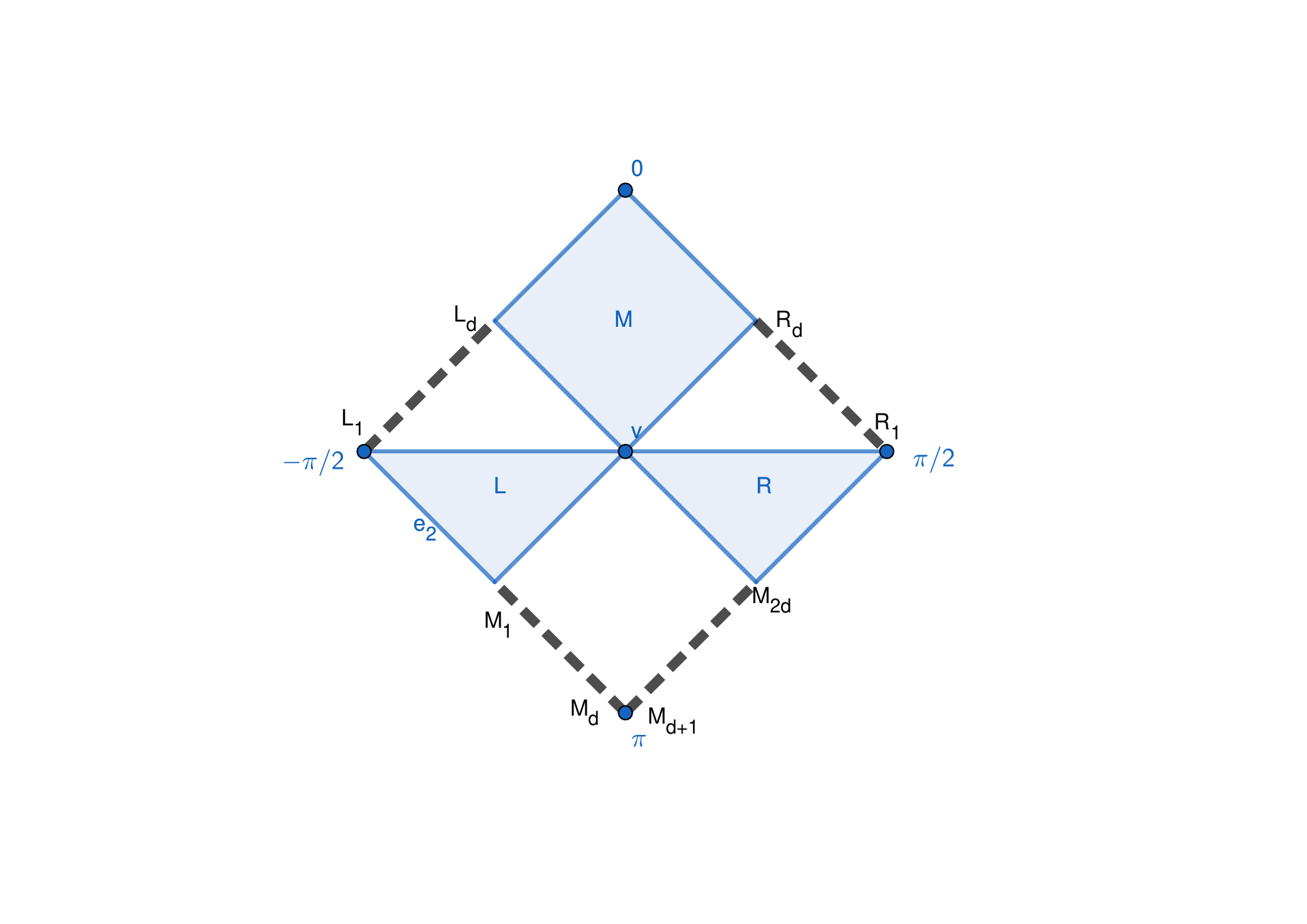}\\
\caption{A joint box for a degree-$d$ vertex, $v$, with left ports, 
$L_1,\ldots,L_d$, right ports, $R_1,\ldots,R_d$,
and middle ports, $M_1,\ldots,M_{2d}$, and free regions,
$L$, $M$, and $R$.
Image from~\cite{Cheng1999DrawingPG}.}
\label{fig:joint}
\end{figure}

The regions of a joint box are of two types, free regions and port regions. 
The free regions are as follows---$M$ is $\pi/4$ on either side of the top corner and $L$ (resp., $R$) is $\pi/4$ below the left (resp., right) corner. 
The port regions are as follows---$M$ is opposite the free $M$ region, 
$L$ (resp., $R$) is between the free $M$ and $L$ (resp., $R$) regions. 
The sides of the square  on port regions have $d$ evenly spaced ports. 
Every edge in the drawing is between a free $M$ region and a port $M$ region, 
a free $L$ region and a port $R$ region, 
or a free $R$ region and a port $L$ region. 
Every edge is drawn as two line segments starting at the endpoints of the edge and meeting at a distinct port on the port region~\cite{Cheng1999DrawingPG}. 

In our construction,
we treat each of the $M$ regions as two regions such that we have $8$ regions which are congruent. However, only one of the two free $M$ regions will have an edge.

\subsection{Drawing a Planar Graph using Cubic \Bezier{} Curves}
\label{drawingdesc}
We take the same embedding of the vertices as given by 
the \texttt{OneBend} algorithm. 
Since these curves can be rotated, translated and reflected without changing the curvature, for any edge
$(A,B)$ which is a 1-bend polyline with the port region at $A$ and 
free region at $B$, without loss of generality, we may assume the following:
$A$ is at the origin, $(0,0)$,
and $B=(b_1,b_2)$ is in the region bounded below by 
the $x$-axis and above by the line $y=x-1$. 
After this transformation the edges through the same port region of $A$ in the drawing obtained by the \texttt{OneBend} algorithm are ported through $(d+1,d+1)+(i,-i)$ for distinct values of $i$ which follow the same order along the vertex and $1\leq i\leq d$. 
In our construction, we replace the polyline edges of
the \texttt{OneBend} algorithm with cubic \Bezier{} curves of the form 
\[
\gamma(t)=Bt^3+3Pt^2(1-t)+3Pt(1-t)^2, 
\]
where $P=(1-k)(1,1)+kQ$, $Q=(1-s)(1,0)+s(3/2,1/2)$, $s=b_2/b_1$, 
and $k=i/(d+1)$. 
The repeated control point, $P$, is a convex combination of $(1,1)$ and $Q$ depending on the parameter, $k$. 
$Q$, in turn, is a convex combination of $(1,0)$ and $(3/2,1/2)$,
depending on the angle $b_2/b_1$ that $B$ makes with the $x$-axis. 
See~\cref{fig:planar-curve}.

\begin{figure}[!t]
    \centering
    \includegraphics[width=5in, trim=0in 2.5in 0in 3in,clip]{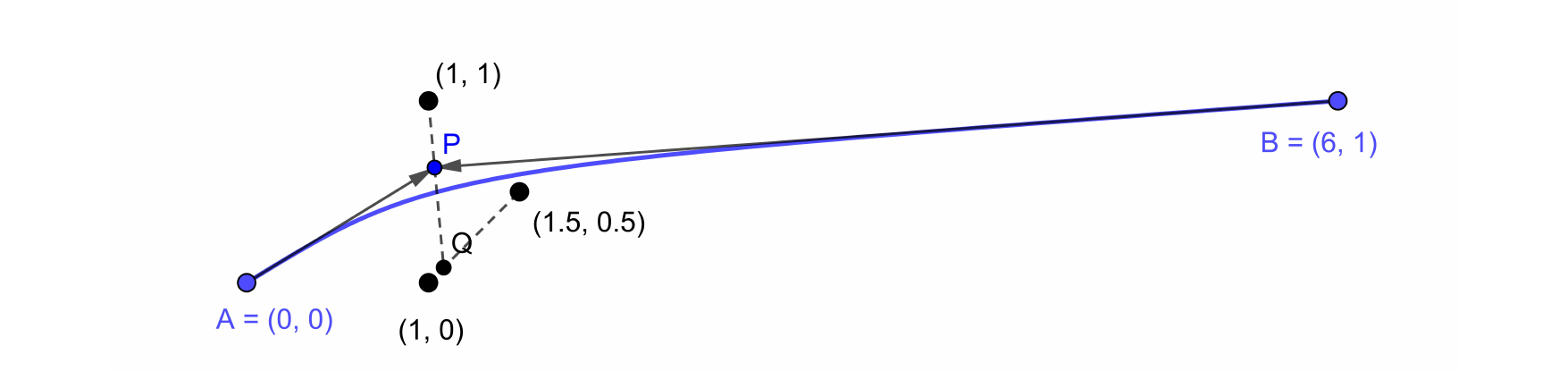}
    \caption{\Bezier{} curve for an edge, $(A,B)$. 
}
    \label{fig:planar-curve}
\end{figure}
%, trim=0in 7in 0in 7in, clip
\subsubsection{Planarity}\label{sec:planarity}
In this subsection,
we show that the cubic \Bezier{} curves representing 
edges with a common endpoint do not cross each other.
First we show that edges that meet at a common endpoint but through different ports do not intersect. We show that an edge through the ported $R$ region does not intersect with an edge through the free $R$ region or the free $M$ region of the same vertex. A similar argument shows that an edge through the ported $L$ region or either half of the ported $M$ region does not intersect with an edge through an adjacent region.

\begin{lemma}
    $\gamma(t)$ lies in the region $x>y>0$ for $t\in(0,1)$.
\end{lemma}
\begin{proof}
The proof follows from the convex hull property of \Bezier{} curves and 
the fact that this region itself is convex. 
The control points except $A$ are within this region. 
Only $A$ could be on the boundary of the region. 
Since all points on the curve except for the endpoints will have 
a positive coefficient for the two intermediate control points, 
these points will lie within the region.
\end{proof}

In other words, an edge ported through the $R$ region will remain
in the wedge of this $R$ region extended to infinity.
Let us next consider whether such an edge will intersect
another one in the free $M$ region.
The control points of any edge through the free $M$ region lie in the region $y\geq x$ and therefore can't intersect with $\gamma$. The control points of any edge through the free $R$ region lie in the region $y\leq 0$ and therefore can't intersect with $\gamma$.

Next we show that two curves through the same ported region do not intersect. 
We will show this for the ported $R$ region but the argument is 
similar for ported $L$ region and both halves of the ported $M$ region.
We actually prove a stronger result, namely, that even if we allowed
for parallel edges with the same two endpoints, they would not cross.

\begin{lemma}\label{changing k}
Suppose $\gamma_1$ and $\gamma_2$ are two edge curves with the 
same endpoints using different values of $k$, $k_1>k_2$, respectively. 
The curves $\gamma_1$ and $\gamma_2$ do not intersect except at the endpoints,
and, except at the endpoints, $\gamma_2$ is above $\gamma_1$.
\end{lemma}
\begin{proof}
    For the common endpoint of $\gamma_1$ and $\gamma_2$,
which has the same $s=b_2/b_1$, the different values of $k$ give different values of $P$ along a line of slope $1-2/s$. For the values of $B$ under consideration, $0\leq s<1$ and hence the slope of this line is less than $-1$.
(See \cref{fig:planar-curve-changing-k}.)
    Consider a rotation such that this line is now perpendicular to the $x$-axis and the two corresponding values of $P$ such that $P_2$ is above $P_1$. All points of $\gamma_2$ would lie directly above the corresponding points of $\gamma_1$. Both curves are monotonic in the horizontal direction since the repeated control point is between the two endpoints in the horizontal direction. The curves do not intersect except at the endpoints. $\gamma_2$ is above $\gamma_1$.
\end{proof}

\begin{figure}[!t]
    \centering
    \includegraphics[width=5in, trim=0.2in 4in 0in 6in, clip]{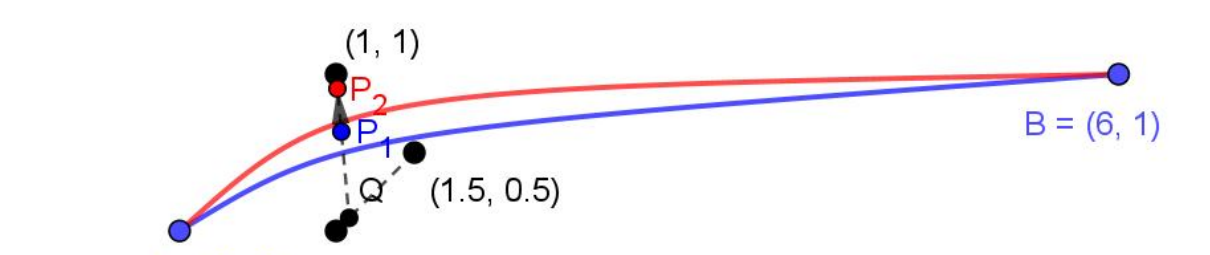}
    \caption{Parallel edges would not cross. As $k$ decreases, all point except the endpoints move up in a direction parallel to $P_1P_2$. The curve itself moves monotonically in the perpendicular direction with $t$.}
    \label{fig:planar-curve-changing-k}
\end{figure}

We next show that two curves with exactly one common endpoint 
and ported through the same $R$ region do not cross.

\begin{lemma}\label{moving B}
    Let $\gamma_1$ and $\gamma_2$ be two curves with the same value of $k$ but different vertices $B$ and $C$ which are endpoints of edges incident on $A$ through the same region in the $\texttt{OneBend}$ algorithm, 
with $B$ coming counterclockwise first. 
The curves intersect only at $A$ and $\gamma_2$ lies 
above $\gamma_1$ everywhere else.
\end{lemma}
\begin{proof}
    Let $P_1$ and $P_2$ be the corresponding values of $P$. 
Since the $1$-bend edge $BC$ has to pass through a pair of ports as described in \cref{onebendalgprop}, the straight line $BC$ has to make an angle between $\pi/4$ and $3\pi/4$ with the $x$-axis. When moving from $B$ to $C$, $s$ increases or remains the same. This means that $P_2$ is the same as $P_1$ or is further away along a line making angle $\pi/4$ with the $x$-axis.(See \cref{fig:planar-curve-changing-B}.) Consider the transformation $\{x,y\}\mapsto\{x-y,x+y\}$. Within this proof, unless noted otherwise, everything will be considered under this transformation. For any value of $P$, the $x$-coordinate is $k$, which is between $0$ and $1$ and the $y$-coordinate is $2-k(1-s)$ which is between $1$ and $2$. $A$ has both coordinates less than those of $P_1$ and $P_2$. $B$ and $C$ have both coordinates greater than those of $P_1$ and $P_2$ respectively. This means that both curves are monotonic along both axes. $C_x\leq B_x$ and $C_y\geq B_y$ with at least one of the inequalities strict. Suppose there are $t_1,t_2$ such that $\gamma_1(t_1)=\gamma_2(t_2)$. $\gamma_2(t_1)_y\geq\gamma_1(t_1)_y=\gamma_2(t_2)_y$ with equality only when $B_y=C_y$. Monotonicity implies that $t_1\geq t_2$. Similarly, $\gamma_2(t_1)_x\leq\gamma_1(t_1)_x=\gamma_2(t_2)_x$ with equality only when $B_x=C_x$. Monotonicity implies that $t_1\leq t_2$. The curves do not intersect unless $B=C$ which is not the case. $\gamma_2$ lies above $\gamma_1$ because $\gamma_2(t)_y\geq\gamma_1(t)_2$ and the curves do not intersect. 
\end{proof}
For two curves through the same ported region, they will have different endpoints and different values of $k$. The curve with the same value of $k$ as the upper edge and the endpoint the same as the lower edge, lies below the curve representing the upper edge and above the curve representing the lower edge. Therefore, the curves representing the upper and lower edges do not intersect.
\begin{figure}[t]
    \centering
    \includegraphics[width=5in, trim=1.8in 7in 0.8in 3in, clip]{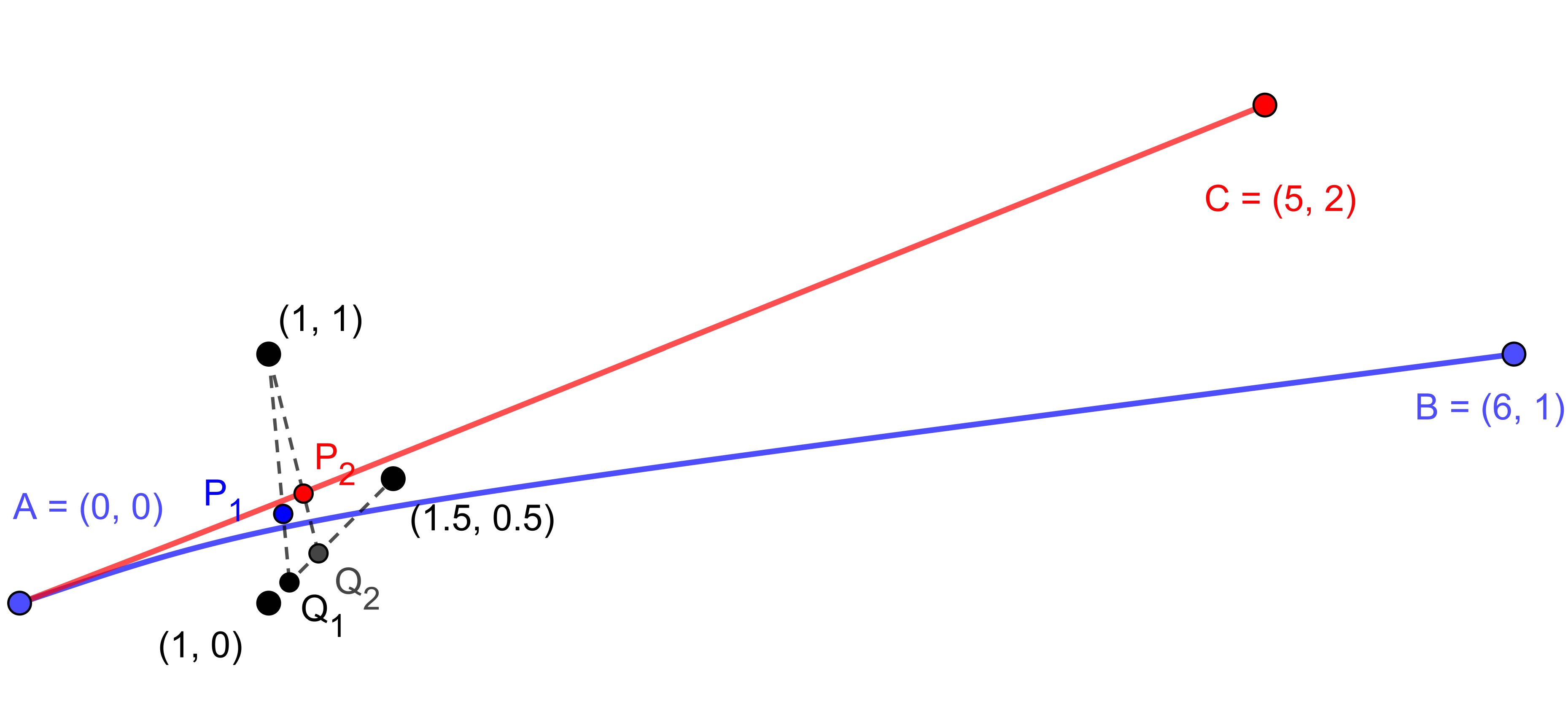}
    \caption{When the endpoint changes from $B$ to $C$ and $k$ remains the same, all points except $A$ move in the direction $P_1P_2$ which is the same as $Q_1Q_2$ or remain where they are. In the perpendicular direction, they move to the left, away from $\gamma_1$.}
    \label{fig:planar-curve-changing-B}
\end{figure}
%\section{Curvature}\label{curvature}

\subsubsection{Angular Resolution}\label{sec:angres}
We next show that for each vertex, $v\in G$, in this drawing,
we achieve angular resolution $\Omega(1/{\rm degree}(v))$.

\begin{lemma}
Let $v\in G$ be a vertex with degree ${\rm degree}(v)$. The angular resolution
of $v$ in our drawing of $G$ using cubic \Bezier{} curves is 
$\Omega(1/{\rm degree}(v))$.
\end{lemma}
\begin{proof}
Without loss of generality, let us assume that the vertex $v$ is at $A=(0,0)$ and has degree $d$. We will consider the angle between any pair of adjacent edges, at least one of which is ported through the $R$ region. The other cases are similar. We first consider the case when both edges are ported through the $R$ region.
The tangent of each \Bezier{} curve for an edge incident to $v$,
at the point of contact is the same direction as the line segment between $A$ and the control point $P$. For a pair of curves with the same pair of endpoints but different values of $P$, that is, $P_1$ and $P_2$ due to different values of $k$ due to different values of $i$, consider the triangle $P_1AP_2$. $P_1P_2$ has length at least $1/(\sqrt2d)$. From the law of sines,
\begin{eqnarray*}
\frac{1/(\sqrt{2}d)}{\sin{\angle P_1AP_2}}=\frac{P_1A}{\sin{\angle P_1P_2A}},\\
\frac{1}{\sqrt{2}d\sin{\angle P_1AP_2}}\leq\frac{\sqrt{10}/2}{1/\sqrt{2}},
\mbox{~~~and}\\
\sin{\angle P_1AP_2}\geq\frac{1}{\sqrt{10}d}.
\end{eqnarray*}
From, the Taylor expansion of $\sin^{-1}$,
\[\angle P_1AP_2\geq \frac{1}{\sqrt{10}d}. \]
Thus,
we have $\Omega(1/d)$ angle between this pair of parallel edges. From \cref{moving B}, we know that the second curve would be above these edges and would have an even larger angle. Note that these curves also leave at least the same angle with the boundaries of ported $R$ region since we use $k=i/(d+1)$ and $i$ ranges from $1$ to $d$ and $0$ and $1$ correspond to the corners. This means that this minimum angle is maintained with any edge incident to $v$ through any other region as well.
\end{proof}

\subsubsection{Curvature}\label{sec:curvature}
We next show that the curvature of our construction is $O(\sqrt{n})$.

\begin{lemma}
    The curvature of $\gamma$ is $O(\sqrt{b_1})\subseteq O(\sqrt{n})$.
\end{lemma}
\begin{proof}
    We denote the curvature by $\kappa$.
By the definition of curvature,
    \[
\frac{\partial (\kappa^2/b_1)}{\partial b_1}=\frac{\begin{matrix}8 (t-1)^2 t^2 \left(k \left(s^2-s+2\right)+2 (s-1)\right)^2 &\\ \left(\begin{matrix}-2 \left(5 b_1^2 \left(s^2+1\right)
   t^4-8 b_1 (s+1) t^3+4 t^2 (b_1 s+b_1-2)+8 t-2\right)&\\+4 k (s-1) (2 t-1) \left(b_1 s t^2+2
   t-1\right)+k^2 \left(s^2-2 s+2\right) (1-2 t)^2\end{matrix}\right)\end{matrix}}{9 \left(\begin{matrix}2 b_1^2 \left(s^2+1\right) t^4-2 k (s-1) (2
   t-1) \left(b_1 s t^2-4 t+2\right)&\\-8 b_1 (s+1) t^3+4 t^2 (b_1 s+b_1+4)+k^2 \left(s^2-2
   s+2\right) (1-2 t)^2-16 t+4\end{matrix}\right)^4} .
\]
In terms of $b_1$, the numerator of 
$\frac{\partial (\kappa^2/b_1)}{\partial b_1}$ is a quadratic polynomial with 
the coefficient of the degree-$2$ term being negative. 
The denominator is a fourth power and hence always non-negative.
    For values of $b_1$, $s$, $k$ and $t$ that we care about, $\frac{\partial (\kappa^2/b_1)}{\partial b_1}$ has at most one root,
\[\frac{-(2 t-1) \left(\begin{matrix}\sqrt{2 k^2 \left(7 s^4-14 s^3+17 s^2-10 s+10\right)+8 k \left(7 s^3-5
   s^2+3 s-5\right)+8 \left(7 s^2+4 s+7\right)}&\\-2 k s^2+2 (k-2) s-4\end{matrix}\right)}{10 \left(s^2+1\right) t^2},
\] 
if $t<0.222$ and no roots otherwise. 
It is negative everywhere except possibly before this root. 
In addition, $\kappa^2/b_1$ decreases as $b_1$ increases after this point for fixed $k$, $s$ and $t$, because this a quadratic with a negative leading coefficient.
At $b_1=4$, the curvature is 
\[
\frac{128 (t-1)^2 t^2 \left(k \left(s^2-s+2\right)+2 (s-1)\right)^2}{9 \left(\begin{matrix}k^2 \left(s^2-2 s+2\right) (1-2 t)^2&\\-4 k
   (s-1) (2 t-1) \left(2 s t^2-2 t+1\right)&\\+4 \left(8 \left(s^2+1\right) t^4-8 (s+1) t^3+4 (s+2) t^2-4
   t+1\right)\end{matrix}\right)^3}.
\] 
This is less than $3$ for all values of $s$, $k$ and $t$ between $0$ and $1$.
    We evaluate $\kappa^2/b_1$ at the possible root of $\frac{\partial (\kappa^2/b_1)}{\partial b_1}$ and get the following:\footnote{We verified 
   this expression, as well as the others in this proof, using 
   Mathematica~\cite{Mathematica}.}
\[
\frac{12500\begin{matrix} \left(s^2+1\right)^2  \left(k \left(s^2-s+2\right)+2 (s-1)\right)^2&\\ \left(\begin{matrix}2 k s^2-2
   (k-2) s+4&\\-\sqrt{\begin{matrix}2 k^2
   \left(7 s^4-14 s^3+17 s^2-10 s+10\right)&\\+8 k \left(7 s^3-5 s^2+3 s-5\right)+8 \left(7 s^2+4 s+7\right)\end{matrix}}\end{matrix}\right)\end{matrix}(t-1)^2}{243 \left(\begin{matrix}4 k^2 \left(2 s^4-4 s^3+7 s^2-5 s+5\right)&\\+k (s-1) \left(\begin{matrix}s \left(\sqrt{\begin{matrix}2 k^2
   \left(7 s^4-14 s^3+17 s^2-10 s+10\right)&\\+8 k \left(7 s^3-5 s^2+3 s-5\right)+8 \left(7 s^2+4 s+7\right)\end{matrix}}-8\right)&\\+32
   s^2+40\end{matrix}\right)&\\+2 \left(\begin{matrix}s \left(\sqrt{\begin{matrix}2 k^2
   \left(7 s^4-14 s^3+17 s^2-10 s+10\right)&\\+8 k \left(7 s^3-5 s^2+3 s-5\right)+8 \left(7 s^2+4 s+7\right)\end{matrix}}-8\right)&\\+\sqrt{\begin{matrix}2 k^2
   \left(7 s^4-14 s^3+17 s^2-10 s+10\right)&\\+8 k \left(7 s^3-5 s^2+3 s-5\right)+8 \left(7 s^2+4 s+7\right)\end{matrix}}+16 s^2+16\end{matrix}\right)\end{matrix}\right)^3(2 t-1)^5}.
\] 
This is a product of two expressions, one independent of $t$ and the other entirely in $t$. The magnitude of the expression in $t$ is bounded by $12$ for all $0\leq t\leq0.222$. The magnitude of the other expression is bounded by $1/128$ for all values of $s$ and $k$ between $0$ and $1$. Putting these together we get that $\kappa^2/b_1$ is bounded by $12/128$.
    The curvature of these curves is $O(\sqrt{b_1})$ and hence $O(\sqrt{n})$ since $b_1$ is $O(n)$.
\end{proof}
  
To sum up, we have the following theorem.

\begin{theorem}
Given an $n$-vertex planar graph, $G$, we can draw $G$ 
in an $O(n)\times O(n)$ grid and $\Omega(1/{\rm degree}(v))$ angular resolution,
for each vertex $v\in G$,
using a single cubic \Bezier{} curve with curvature $O(\sqrt{n})$
per edge in $O(n)$ time.
\end{theorem}

\section{Conclusion}
In this paper, we have studied methods for drawing
1-planar and planar graphs using cubic \Bezier{} curves with
bounded curvature.
Possible directions for future work and open problems include the following:
\begin{itemize}
\item
Can the curvature for drawing an $n$-vertex planar graph using 
a single \Bezier{} curve for each edge be improved from
$O(\sqrt{n})$ while still maintaining an angular resolution of
$\Omega(1/{\rm degree}(v))$, for each vertex $v$ in the drawing?
\item
Our algorithm used to produce a RAC drawing of an $n$-vertex 1-planar graph
$G$,
given a combinatorial 1-planar drawing of $G$, 
is based on the recursive construction of 
Bekos, Didimo, Liotta, Mehrabi, and Montecchiani~\cite{BEKOS201748}.
This allows us to achieve bounded curvature for each edge in the
drawing, but it does not give us a bound on the curvature
in terms of $n$. Is it possible to achieve such a bound?
\item
Our algorithm used to produce a RAC drawing of an $n$-vertex 1-planar graph
uses ``S''-shaped \Bezier{} curves. 
Can the same result be achieved with ``C''-shaped
\Bezier{} curves, e.g., quadratic \Bezier{} curves, which are arguably
more aesthetically pleasing?
\item
As mentioned above,
Lombardi drawings are drawings
where edges are drawn using circular arcs, 
e.g., see~\cite{roman12,duncan2012lombardi,duncan2012planar,eppstein2012planar,eppstein12,kindermann2017lombardi}, 
but circular arcs are not \Bezier{} curves.
Can every graph with a Lombardi drawing also be drawn with the same edge
crossings using a single \Bezier{} curve of bounded curvature for
each edge?
\end{itemize}

\bibliographystyle{plainurl}
\bibliography{ref}
\end{document}